\documentclass{article}

\usepackage{authblk}
\usepackage{amsmath, amssymb, amsthm}
\usepackage{mathtools}
\usepackage{clrscode4e}
\usepackage{enumerate}
\usepackage[parfill]{parskip}
\allowdisplaybreaks

\newtheorem{theorem}{Theorem}
\newtheorem{lemma}{Lemma}
\newtheorem{definition}{Definition}

\newtheorem{corollary}{Corollary}

\title{Approximating Dense Max 2-CSPs}
\date{\today}

\author[1]{Pasin Manurangsi\thanks{Part of this work was completed while the author was at Massachusetts Institute of Technology.}}
\author[2]{Dana Moshkovitz\thanks{This material is based upon work supported by the National Science Foundation under Grant Number 1218547.}}
\affil[1]{Dropbox, Inc.\\
  San Francisco, CA 94107, USA\\
  \texttt{pasin@dropbox.com}}
\affil[2]{Massachusetts Institute of Technology \\
  Cambridge, MA 02139, USA\\
  \texttt{dmoshkov@mit.edu}}

\begin{document}
\maketitle


\begin{abstract}
In this paper, we present a polynomial-time algorithm that approximates sufficiently high-value {\sc Max 2-CSP}s on sufficiently dense graphs to within $O(N^{\varepsilon})$ approximation ratio for \emph{any} constant $\varepsilon > 0$. Using this algorithm, we also achieve similar results for free games, projection games on sufficiently dense random graphs, and the {\sc Densest $k$-Subgraph} problem with sufficiently dense optimal solution. Note, however, that algorithms with similar guarantees to the last algorithm were in fact discovered prior to our work by Feige et al. and Suzuki and Tokuyama.

In addition, our idea for the above algorithms yields the following by-product: a quasi-polynomial time approximation scheme (QPTAS) for satisfiable dense {\sc Max 2-CSP}s with better running time than the known algorithms. \\

{\bf Keywords:} {\sc Max 2-CSP}, Dense Graphs, {\sc Densest $k$-Subgraph}, QPTAS, Free Games, Projection Games
\end{abstract}

\section{Introduction}

Maximum constraint satisfaction problem ({\sc Max CSP}) is a problem of great interest in approximation algorithms since it encapsulates many natural optimization problems; for instance, {\sc Max $k$-SAT}, {\sc Max-Cut}, {\sc Max-DiCut}, {\sc Max $k$-Lin}, projection games, and unique games are all families of {\sc Max CSP}s. In {\sc Max CSP}, the input is a set of variables, an alphabet set, and a collection of constraints. Each constraint's domain consists of all the possible assignments to a subset of variables. The goal is to find an assignment to all the variables that satisfies as many constraints as possible.

In this paper, our main focus is on the case where each constraint depends on exactly $k=2$ variables and the alphabet size is large. This case is intensively researched in hardness of approximation and multi-prover games.

For Max $2$-CSP with large alphabet size, the best known polynomial-time approximation algorithm, due to Charikar et al.~\cite{CHK}, achieves an approximation ratio of $O((nq)^{1/3})$ where $n$ is the number of variables and $q$ is the alphabet size. On the other hand, it is known that, there is no polynomial-time $2^{\log^{1-\delta} (nq)}$-approximation algorithm for {\sc Max $2$-CSP} unless NP $\not\subseteq$ DTIME$(n^{polylog(n)})$~\cite{Raz}. Moreover, it is believed that, for some constant $c > 0$, no polynomial-time $O((nq)^c)$-approximation algorithm exists for projection games, a family of {\sc Max $2$-CSP} we shall introduce later, unless P = NP~\cite{M}. This is also known as the Projection Games Conjecture (PGC). As a result, if the PGC holds, one must study special cases in order to go beyond polynomial approximation ratio for {\sc Max 2-CSP}.

One such special case that has been particularly fruitful is dense {\sc Max $2$-CSP} where density is measured according to number of constraints, i.e., an instance is $\delta$-dense if there are $\delta n^2$ constraints. Note that, for convenience, we always assume that there is at most one constraint on a pair of variables. In other words, we form a simple graph by letting vertices represent the variables and edges represent the constraints. This is the interpretation that we will use throughout the paper. According to this view, $\delta n$ is the average degree of the graph.

In 1995, Arora, Karger and Karpinski~\cite{AKK95} invented a polynomial-time approximation scheme (PTAS) for dense Max $2$-CSP when the density $\delta$ and alphabet size $q$ are constants. More specifically, for any constant $\varepsilon > 0$, the algorithm achieves an approximation ratio $1 + \varepsilon$ and runs in time $O(n^{1/\varepsilon^2})$. Unfortunately, the running time becomes quasi-polynomial time when $q$ is not constant.

Another line of development of such PTASs centers around subsampling technique (e.g. ~\cite{AIM, Alon:2003:RSA:963875.963877, BMHS11}). In summary, these algorithms function by randomly sampling the variables according to some distribution and performing an exhaustive search on the induced instance. Since the sampled set of variables is not too large, the running time is not exponential. However, none of these algorithm achieves polynomial running time for large alphabets. In particular, all of them are stuck at quasi-polynomial running time.

Since none of these algorithms runs in polynomial time for large alphabet, a natural and intriguing question is how good a polynomial-time approximation algorithm can be for dense {\sc Max 2-CSP}s. In this paper, we partially answer this question by providing a polynomial-time approximation algorithm for dense high-value {\sc Max 2-CSP}s that achieves $O((nq)^\varepsilon)$ approximation ratio for any constant $\varepsilon > 0$. Moreover, our technique also helps us come up with a quasi-polynomial time approximation scheme for satisfiable {\sc Max 2-CSP}s with running time asymptotically better than that those from~\cite{AIM, Alon:2003:RSA:963875.963877, AKK95, BMHS11}.

The central idea of our technique is a trade-off between two different approaches: greedy assignment algorithm and ``choice reduction'' algorithm. In summary, either a simple greedy algorithm produces an assignment that satisfies many constraints or, by assigning an assignment to just one variable, we can reduce the number of optimal assignment candidates of other variables significantly. The latter is what we call the choice reduction algorithm. By applying this argument repeatedly, either one of the greedy assignments gives a high-value assignment, or we are left with only few candidate labels for each variable. In the latter case, we can then just pick a greedy assignment at the end.

Not only that our technique is useful for {\sc Max 2-CSP}, we are able to obtain approximation algorithms for other problems in dense settings as well. The first such problem is free games, which can be defined simply as {\sc Max $2$-CSP} on balanced complete bipartite graphs. While free games have been studied extensively in the context of parallel repetition~\cite{BRR+09, S13} and as basis for complexity and hardness results~\cite{AIM, BKW15}, the algorithm aspect of it has not been researched as much. In fact, apart from the aforementioned algorithms for dense {\sc Max 2-CSP} that also works for free games, we are aware of only two approximation algorithms, by Aaronson et al.~\cite{AIM} and by Brandao and Harrow~\cite{BH13}, specifically developed for free games. Similar to the subsampling lemmas, these two algorithms are PTASs when $q$ is constant but, when $q$ is large, the running times become quasi-polynomial. Interestingly, our result for dense {\sc Max 2-CSP} directly yields a polynomial-time algorithm that can approximate free games within $O((nq)^\varepsilon)$ factor for any constant $\varepsilon > 0$, which may be the first non-trivial approximation algorithm for free games with such running time.

Secondly, our idea is also applicable for projection games. The projection games problem (also known as {\sc Label Cover}) is {\sc Max 2-CSP} on a bipartite graph where, for each assignment to a left vertex of an edge, there is exactly one satisfiable assignment to the other endpoint of the edge. {\sc Label Cover} is of great significance in the field of hardness of approximation since almost all NP-hardness of approximation results known today are reduced from the NP-hardness of approximation of projection games (e.g.~\cite{BGS,Has97}).

The current best polynomial-time approximation algorithm for satisfiable projection games is the authors' with $O((nq)^{1/4})$ ratio~\cite{MM13}. Moreover, as mentioned earlier, if the PGC is true, then, in polynomial time, approximating {\sc Label Cover} beyond some polynomial ratio is unlikely. In this paper, we exceed this bound on random balanced bipartite graphs with sufficiently high density by proving that, in polynomial time, one can approximate satisfiable projection games on such graphs to within $O((nq)^\varepsilon)$ factor for any constant $\varepsilon > 0$.

Finally, we show a similar result for {\sc Densest $k$-Subgraph}, the problem of finding a size-$k$ subgraph of a given graph that contains as many edges as possible. Finding best polynomial-time approximation algorithm for {\sc Densest $k$-Subgraph}({\sc D$k$S}) is an open question in the field of approximation algorithms. Currently, the best known algorithm for {\sc D$k$S} achieves an approximation ratio of $O(n^{1/4+\varepsilon})$ for any constant $\varepsilon > 0$~\cite{BCCFV}. On the other hand, however, we only know that there is no PTAS for {\sc D$k$S} unless P=NP~\cite{Khot04}.

Even though {\sc Densest $k$-Subgraph} on general graphs remains open, the problem is better understood in some dense settings. More specifically, Arora et al.~\cite{AKK95} provided a PTAS for the problem when the given graph is dense and $k = \Omega(N)$ where $N$ is the number of vertices of the given graph. Later, Feige et al.~\cite{FPK01} and Suzuki and Tokuyama~\cite{ST05} showed that, if we only know that the optimal solution is sufficiently dense, we can still approximate the solution to within any polynomial ratio in polynomial time. Using our approximation algorithm for dense {\sc Max 2-CSP}, we are able to construct a polynomial-time algorithm for {\sc Densest $k$-Subgraph} with similar conditions and guarantees as those of the algorithms from~\cite{FPK01} and~\cite{ST05}.

The theorems we prove in this paper are stated in Section~\ref{s:results} after appropriate preliminaries in the next section.

\section{Preliminaries and Notation}

In this section, we formally define the problems we focus on and the notation we use throughout the paper. First, to avoid confusion, let us state the definition of approximation ratio for the purpose of this paper. \\

\begin{definition}
  An approximation algorithm for a maximization problem is said to have an approximation ratio $\alpha$ if the output of the algorithm is at least $1/\alpha$ times the optimal solution.
\end{definition}

Note here that the approximation ratio as defined above is always at least one.

Next, before we define our problems, we review the standard notation of density of a graph. \\

\begin{definition}
  A simple undirected graph $G = (V, E)$ is defined to be of density $|E|/|V|^2$.
\end{definition}

Moreover, for a graph $G$ and a vertex $u$, we use $\Gamma^G(u)$ to denote the set of neighbors of $u$ in $G$. We also define $\Gamma_2^G(u)$ to denote the set of neighbors of neighbors of $u$ in $G$, i.e., $\Gamma_2^G(u) = \Gamma^G(\Gamma^G(u))$. When it is unambiguous, we will leave out $G$ and simply write $\Gamma(u)$ or $\Gamma_2(u)$.

Now, we will define the problems starting with {\sc Max 2-CSP}. \\

\begin{definition}
An instance $(q, V, E, \{C_e\}_{e \in E})$ of {\sc Max $2$-CSP} consists of
\begin{itemize}
  \item a simple undirected graph $(V, E)$, and
  \item for each edge $e = (u, v) \in E$, a constraint (or constraint) $C_e: [q]^2 \to \{0, 1\}$ where $[q]$ denotes $\{1, 2, \dots, q\}$.
\end{itemize}
The goal is to find an assignment (solution) $\varphi: V \to [q]$ that maximizes the number of constraints $C_e$'s that are satisfied, i.e. $C_{(u, v)}(\varphi(u), \varphi(v)) = 1$. In other words, find an assignment $\varphi: \{x_1, \dots, x_n\} \to [q]$ that maximizes $\sum_{(u, v) \in E} C_{(u, v)}(\varphi(u), \varphi(v))$. The value of an assignment is defined as the fraction of edges satisfied by it and the value of an instance is defined as the value of the optimal assignment.
\end{definition}

A {\sc Max $2$-CSP} instance $(q, V, E, \{C_e\}_{e \in E})$ is called $\delta$-dense if the graph $(V, E)$ is $\delta$-dense. Throughout the paper, we use $n$ to denote the number of vertices (variables) $|V|$ and $N$ to denote $nq$, which can be viewed as the size of the problem.

Free games and projection games are specific classes of {\sc Max 2-CSP}, which can be defined as follows. Note that $n, N$, density and value are defined in a similar fashion for free games and projection games as well. \\

\begin{definition} A free game $(q, A, B, \{C_{a, b}\}_{(a, b) \in A \times B})$ consists of
\begin{itemize}
  \item Two sets $A, B$ of equal size, and
  \item for $(a, b) \in A \times B$, a constraint $C_{a, b} : [q]^2 \to \{0, 1\}$.
\end{itemize}
The goal is to find an assignment $\varphi : A \cup B \to [q]$ that maximizes the number of edges $(a, b) \in A \times B$ that are satisfied, i.e., $C_{a, b}(\varphi(a), \varphi(b)) = 1$. \\
\end{definition}

\begin{definition}
A projection game $(q, A, B, E, \{\pi_e\}_{e \in E})$ consists of
\begin{itemize}
  \item a simple bipartite graph $(A, B, E)$, and
  \item for each edge $e = (a, b) \in E$, a ``projection'' $\pi_e : [q] \to [q].$
\end{itemize}
The goal is to find an assignment to the vertices $\varphi : A \cup B \to [q]$ that maximizes the number of edges $e = (a, b)$ that are satisfied, i.e., $\pi_e(\varphi(a)) = \varphi(b)$.
\end{definition}

Both free games and projection games can be viewed as special cases of {\sc Max 2-CSP}. More specifically, free games are simply Max 2-CSPs on complete balanced bipartite graphs.

For projection games, one can view $\pi_e$ as a constraint $C_e: [q]^2 \to \{0, 1\}$ where $C_e(\sigma_u, \sigma_v) = 1$ if and only if $\pi_e(\sigma_u) = \sigma_v$. In other words, projection game is {\sc Max 2-CSP} on bipartite graph where an assignment to the endpoint in $A$ of an edge determines the assignment to the endpoint in $B$.

For convenience, we will define the notation of ``optimal assignment'' for {\sc Max 2-CSP} intuitively as follows. \\

\begin{definition}
For a {\sc Max 2-CSP} instance $(q, V, E, \{C_e\}_{e \in E})$, for each vertex $u \in V$, let $\sigma_u^{OPT}$ be the assignment to $u$ in an assignment to vertices that satisfies maximum number of edges, i.e., $\varphi(u) = \sigma_u^{OPT}$ is the assignment that maximizes $\sum_{(u, v) \in E} C_{(u, v)}(\varphi(u), \varphi(v))$. In short, we will sometimes refer to this as ``the optimal assignment''.
\end{definition}

Note that since projection games and free games are families of {\sc Max 2-CSP}, the above definition also carries over when we discuss them.

Lastly, we define {\sc Densest $k$-Subgraph}. \\

\begin{definition}
In the {\sc Densest $k$-Subgraph} problem, the input is a simple graph $G = (V, E)$ of $N = |V|$ vertices. The goal is to find a subgraph of size $k$ that contain maximum number of edges.
\end{definition}

\section{Summary of Results}\label{s:results}

We are finally ready to describe our results and how they relate to the previous results. We will start with the main theorem on approximating high-value dense {\sc Max 2-CSP}. \\

\begin{theorem}[Main Theorem] \label{thm:main}
  For every constant $\gamma > 0$, there exists a polynomial-time algorithm that, given a $\delta$-dense {\sc Max 2-CSP} instance of value $\lambda$, produces an assignment of value $\Omega((\delta \lambda)^{O(1/\gamma)}N^{-\gamma})$ for the instance.
\end{theorem}

Note that, when $\delta, \lambda = N^{-o(1)}$, by choosing $\gamma < \varepsilon$, the algorithm can achieve $O((nq)^\varepsilon)$ approximation ratio for any constant $\varepsilon > 0$.

Since every free game is $1/2$-dense, Theorem~\ref{thm:main} immediately implies the following corollary. \\

\begin{corollary} \label{cor:free-game}
  For every constant $\gamma > 0$, there exists a polynomial-time algorithm that, given a free game of value $\lambda$, produces an assignment of value $\Omega(\lambda^{O(1/\gamma)}N^{-\gamma})$ for the instance.
\end{corollary}

Again, note that when $\lambda = N^{-o(1)}$, the algorithm can achieve $O((nq)^\varepsilon)$ approximation ratio for any constant $\varepsilon > 0$.

The next result is a similar algorithm for projection games on sufficiently dense random graphs as stated below. \\

\begin{theorem} \label{thm:dense-random-proj}
  For every constant $\gamma > 0$, there exists a polynomial-time algorithm that, given a satisfiable projection game on a random bipartite graph $(A, B, E) \sim \mathcal{G}(n/2, n/2, p)$ for any $p \geq 10\sqrt{\log n / n}$, produces an assignment of value $\Omega(N^{-\gamma})$ for the instance with probability $1 - o(1)$.
\end{theorem}

Note that $\mathcal{G}(n/2, n/2, p)$ is defined in Erd\H{o}s-R\'{e}nyi fashion, i.e., the graph contains $n/2$ vertices on each side and, each pair of left and right vertices is included as an edge with probability $p$ independently.

In addition, it is worth noting here that the required density for projection games is much lower than that of {\sc Max 2-CSP}; our {\sc Max 2-CSP} algorithm requires the degree to be $\Omega(n/N^{-o(1)})$ whereas the projection games algorithm requires only $\tilde{\Omega}(\sqrt{n})$.

As stated earlier, we are unaware of any non-trivial polynomial-time algorithm for dense {\sc Max 2-CSP}, free games, or projection games on dense random graphs prior to our algorithm.

Next, we state our analogous result for {\sc Densest $k$-Subgraph}. \\

\begin{corollary} \label{cor:dks}
  For every constant $\gamma > 0$, there exists a polynomial-time algorithm that, given a graph $G = (V, E)$ on $N$ vertices such that its densest subgraph with $k$ vertices is $\delta$-dense, produces a subgraph of $k$ vertices that is $\Omega(\delta^{O(1/\gamma)}N^{-\gamma})$-dense with high probability.
\end{corollary}

Note that the density condition is on the optimal solution, not the given graph $G$. The condition and the algorithm are exactly the same as that of~\cite{FPK01} and~\cite{ST05}. However, the techniques are substantially different. While~\cite{FPK01} deals combinatorially directly with the given graph $G$ and~\cite{ST05} employs subsampling technique, we simply use our algorithm from Theorem~\ref{thm:main} together with a simple reduction from {\sc Densest $k$-Subgraph} to {\sc Max 2-CSP} due to Charikar et al.~\cite{CHK}.

Lastly, we also give a quasi-polynomial time approximation scheme for satisfiable dense {\sc Max 2-CSP} as described formally below. \\

\begin{corollary}[QPTAS for Dense {Max 2-CSP}] \label{cor:qptas-dense}
  For any $1 \geq \varepsilon > 0$, there exists an $(1 + \varepsilon)$-approximation algorithm for satisfiable $\delta$-dense {\sc Max 2-CSP} that runs in time $N^{O(\varepsilon^{-1}\delta^{-1}\log N)}$.
\end{corollary}

Comparing to the known algorithms, our QPTAS runs faster than QPTASs from~\cite{Alon:2003:RSA:963875.963877, AKK95, BMHS11}, each of which takes at least $N^{O(\varepsilon^{-2}\delta^{-1}\log N)}$ time. However, while our algorithm works only for satisfiable instances, the mentioned algorithms work for unsatisfiable instances as well but with an additive error of $\varepsilon$ in value instead of the usual multiplicative guarantee of $(1 + \varepsilon)$.

\section{Proof of The Main Theorem}

In this section, we prove the main theorem. In order to do so, we will first show that we do not have to worry about the density $\delta$ at all, i.e., it is enough for us to prove the following lemma. \\

\begin{lemma} \label{lem:main}
  For every $\gamma > 0$, there exists a polynomial-time algorithm that, given a free game $(q, A, B, \{C_{(a, b)}\}_{(a, b) \in A \times B})$ of value $\lambda'$, produces an assignment of value $\lambda'^{O(1/\gamma)}q^{-\gamma}$ for the instance.
\end{lemma}

The proof of the main theorem based on the lemma above is shown below.

\begin{proof}[Proof of Theorem~\ref{thm:main} based on Lemma~\ref{lem:main}]
  The proof is based on putting in ``dummy edges'' where the constraints are always false regardless of the assignment to make the game more dense. More specifically, given a {\sc Max 2-CSP} instance $(q, V, E, \{C_e\}_{e \in E})$ of value $\lambda$ and density $\delta$, we construct a free game $(q', A, B, \{C'_{(a, b)}\}_{(a, b) \in A \times B})$ as follows:
  \begin{itemize}
    \item Let $A, B$ be copies of $V$ and let $q' = q$.
    \item For each $a \in A$ and $b \in B$, let $C'_{(a, b)} = C_{(a, b)}$ if $(a, b) \in E$. Otherwise, let $C'_{(a, b)} := 0$.
  \end{itemize}

  It is not hard to see that, if we assign the optimal assignment of the original instance to the free game, then $\delta \lambda n^2$ edges are satisfied where $n = |V|$. In other words, the value of the free game is at least $\delta \lambda$. Thus, from Lemma~\ref{lem:main}, for any constant $\gamma$, we can find an assignment $\varphi': A \cup B \to [q']$ of value at least $(\delta \lambda)^{O(1/\gamma)}q^{-\gamma}$ for the free game.

  We create an assignment $\varphi: V \to [q]$ based on $\varphi'$ as follows. For each vertex $v \in V$, let $a_v \in A$ and $b_v \in B$ be the vertices corresponding to $v$ in the free game. Set $\varphi(v)$ to be either $\varphi'(a_v)$ or $\varphi'(b_v)$ with equal probability.

  From the above construction, the expected number of edges satisfied by $\varphi$ in the {\sc Max 2-CSP} instance is
  \begin{align*}
  	\mathbb{E}\left[\sum_{(u, v) \in E} C_{(u, v)}(\varphi(u), \varphi(v))\right] &= \sum_{(u, v) \in E} \mathbb{E}\left[C_{(u, v)}(\varphi(u), \varphi(v))\right] \\
  	(\text{From our choice of } \varphi)&= \sum_{(u, v) \in E} \frac{1}{4} \left(\sum_{\sigma_u \in \{\varphi'(a_u), \varphi'(b_u)\}} \sum_{\sigma_v \in \{\varphi'(a_v), \varphi'(b_v)\}} C_{(u, v)}(\sigma_u, \sigma_v)\right) \\
  	&\geq \sum_{(u, v) \in E} \frac{1}{4} \left(C_{(u, v)}(\varphi'(a_u), \varphi'(b_v)) + C_{(u, v)}(\varphi'(b_u), \varphi'(a_v))\right) \\
  	&= \frac{1}{4} \sum_{(u, v) \in E} \left(C_{(u, v)}(\varphi'(a_u), \varphi'(b_v)) + C_{(u, v)}(\varphi'(b_u), \varphi'(a_v))\right) \\
  	(\text{From definition of } C') &= \frac{1}{4} \sum_{(a, b) \in A \times B} C'_{(a, b)}(\varphi'(a), \varphi'(b)). \\
  \end{align*}

  Observe that $$\sum_{(a, b) \in A \times B} C'_{(a, b)}(\varphi'(a), \varphi'(b))$$ is the value of $\varphi'$ with respect to the free game, which is at least $(\delta \lambda)^{O(1/\gamma)}q^{-\gamma}$. As a result, we can conclude that $\varphi$ is of expected value at least $\frac{1}{4}(\delta \lambda)^{O(1/\gamma)}q^{-\gamma} = \Omega((\delta \lambda)^{O(1/\gamma)}N^{-\gamma})$ with respect to the instance $(q, V, E, \{C_e\}_{e \in E})$.

  Lastly, we note that while the algorithm above is non-deterministic, the standard derandomization technique via conditional probability can be employed to make the algorithm deterministic without affecting the guarantee on the value of $\varphi$, which completes our proof for the main theorem.
\end{proof}

Now, we finally give the proof for Lemma~\ref{lem:main}. As mentioned in the introduction, the main idea of the proof is a trade-off between the greedy algorithm and the choice reduction algorithm. In other words, either the greedy assignment has high value, or we can reduce the number of candidates of the optimal assignment for many variables significantly by assigning only one variable. This argument needs to be applied multiple times to arrive at the result; the more variables we iterate on, the better guarantee we get on the output assignment value.

For the purpose of analysis, we will define our algorithm recursively and use induction to show that the output assignment meets the desired criteria.

\begin{proof}[Proof of Lemma~\ref{lem:main}]
First, let us define notation that we will use throughout the proof. For a free game $(q, A, B, \{C_{(a, b)}\}_{(a, b) \in A \times B})$, define $E^{OPT}$ to be the set of edges satisfied by $\{\sigma^{OPT}_u\}_{u \in V}$. In other words, $E^{OPT} = \{(u, v) \in E \mid C_{(u, v)}(\sigma^{OPT}_u, \sigma^{OPT}_v) = 1\}$. We also define $\Gamma^{OPT}(u)$ to be the neighborhood of $u$ with respect to $(V, E^{OPT})$ and let $d^{OPT}_u$ be the degree of $u$ in $(V, E^{OPT})$, i.e., $d^{OPT}_u = |\Gamma^{OPT}(u)|$. In addition, let $n' = n/2$ be the size of $A$ and $B$.

We will prove the lemma by induction. Let $P(i)$ represent the following statement: there exists an $O\left((nq)^{2i}\right)$-time algorithm {\sc Approx-FreeGame$_i$}($q, A, B,$ $\{C_{(a, b)}\}_{(a, b) \in A \times B}, \{S_b\}_{b \in B}$) that takes in a free game instance $(q, A, B, \{C_{(a, b)}\}_{(a, b) \in A \times B})$ of value $\lambda'$ and a reduced alphabet set $S_b$ for every $b \in B$, and produces an assignment that satisfies at least $$n'\left(\sum_{b \in B} \left(\frac{d^{OPT}_b}{n'}\right)^\frac{i+1}{2}\left(\frac{1}{|S_b|}\right)^\frac{1}{i}1_{\sigma_b^{OPT} \in S_b}\right)$$ edges. Note here that $1_{\sigma_b^{OPT} \in S_b}$ denotes an indicator variable for whether $\sigma_b^{OPT} \in S_b$. Moreover, for convenience, we use the expression $\left(1/|S_b|\right)^\frac{1}{i}1_{\sigma_b^{OPT} \in S_b}$ to be represent zero when $S_b = \emptyset$.

Before we proceed to the induction, let us note why $P(i)$ implies the lemma. By setting $i = \lceil 1/\gamma \rceil$ and $S_b = [q]$ for every $b \in B$, since $\sigma_b^{OPT} \in S_b$ for every $b \in B$, the number of edges satisfied by the output assignment of the algorithm in $P(i)$ is at least
\begin{align*}
  n' \sum_{b \in B} \left(\frac{d^{OPT}_b}{n'}\right)^\frac{i+1}{2}\left(\frac{1}{q}\right)^\frac{1}{i} &= n' \frac{1}{q^{1/i}} \left(\sum_{b \in B} \left(\frac{d^{OPT}_b}{n'}\right)^\frac{i+1}{2}\right) \\
  (\text{From H\"{o}lder's inequality}) &\geq \frac{(n')^2}{q^{1/i}} \left(\frac{1}{n'}\sum_{b \in B} \frac{d^{OPT}_b}{n'}\right)^\frac{i+1}{2} \\
  &= \frac{(n')^2}{q^{1/i}} \left(\frac{|E^{OPT}|}{(n')^2}\right)^\frac{i+1}{2} \\
  (\text{Since } |E^{OPT}|/(n')^2 \text{ is the value of the instance}) &= \frac{(n')^2}{q^{1/i}} \left(\lambda'\right)^\frac{i+1}{2} \\
  (\text{From our choice of } i) &\geq (n')^2 \frac{\lambda'^{O(1/\gamma)}}{q^{\gamma}},
\end{align*}
which is the statement of the lemma.

Now, we finally show that $P(i)$ is true for every $i \in \mathbb{N}$ by induction.

{\em Base Case.} The algorithm {\sc Approx-FreeGame$_1$}($q, A, B, \{C_{(a, b)}\}_{(a, b) \in A \times B}, \{S_b\}_{b \in B}$) is a greedy algorithm that works as follows:
\begin{enumerate}
  \item For each $a \in A$, assign $\sigma^*_a \in S_a$ that maximizes $\sum_{b \in B} \frac{1}{|S_b|}\left(\sum_{\sigma_b \in S_b} C_{(a, b)}(\sigma_a, \sigma_b)\right)$ to it. \label{step:free-greedy}
  \item For each $b \in B$, assign $\sigma^*_b \in S_b$ that maximizes the number of edges satisfied, i.e., $\sum_{a \in A} C_{(a, b)}(\sigma^*_a, \sigma_b)$, to it.
\end{enumerate}

It is obvious that the algorithm runs in $O(n^2q^2)$ time as desired.

Next, we need to show that the algorithm gives an assignment that satisfies at least $$n'\left(\sum_{b \in B} \left(\frac{d^{OPT}_b}{n'}\right)\left(\frac{1}{|S_b|}\right)1_{\sigma_b^{OPT} \in S_b}\right) = \sum_{b \in B} \frac{d^{OPT}_b}{|S_b|}1_{\sigma_b^{OPT} \in S_b}$$ edges.

To prove this, observe that, from our choice of $\sigma^*_b$, the number of satisfied edges by the output assignment can be bounded as follows.
\begin{align*}
  \sum_{b \in B} \sum_{a \in A} C_{(a, b)}(\sigma^*_a, \sigma^*_b) &\geq \sum_{b \in B} \frac{1}{|S_b|} \sum_{\sigma_b \in S_b} \left(\sum_{a \in A} C_{(a, b)}(\sigma^*_a, \sigma_b)\right) \\
  &= \sum_{a \in A} \sum_{b \in B} \frac{1}{|S_b|} \left(\sum_{\sigma_b \in S_b} C_{(a, b)}(\sigma^*_a, \sigma_b)\right) \\
  (\text{From our choice of } \sigma_a^*) &\geq \sum_{a \in A} \sum_{b \in B} \frac{1}{|S_b|} \left(\sum_{\sigma_b \in S_b} C_{(a, b)}(\sigma^{OPT}_a, \sigma_b)\right) \\
  &\geq \sum_{a \in A} \sum_{b \in B} \frac{1}{|S_b|} C_{(a, b)}(\sigma^{OPT}_a, \sigma^{OPT}_b) 1_{\sigma_b^{OPT} \in S_b} \\
  &= \sum_{b \in B} \sum_{a \in A} \frac{1}{|S_b|} C_{(a, b)}(\sigma^{OPT}_a, \sigma^{OPT}_b) 1_{\sigma_b^{OPT} \in S_b} \\
  (\text{From definition of } d^{OPT}_b)&= \sum_{b \in B} \frac{1}{|S_b|} d^{OPT}_b 1_{\sigma_b^{OPT} \in S_b} \\
  &= \sum_{b \in B} \frac{d^{OPT}_b}{|S_b|} 1_{\sigma_b^{OPT} \in S_b}.
\end{align*}

Thus, we can conclude that $P(1)$ is true.

{\em Inductive Step.} Let $j$ be any positive integer. Suppose that $P(j)$ holds.

We will now describe {\sc Approx-FreeGame$_{j+1}$} based on {\sc Approx-FreeGame$_j$} as follows.

\begin{enumerate}
\item For each $a \in A$ and $\sigma_a \in S_a$, do the following:
  \begin{enumerate}
  \item For each $b \in B$, compute $S^{a, \sigma_a}_b = \{\sigma_b \in S_b \mid C_{(a, b)}(\sigma_a, \sigma_b) = 1\}$.
  \item Call {\sc Approx-FreeGame$_{j}$}($q, A, B, \{C_{(a, b)}\}_{(a, b) \in A \times B}, \{S^{a, \sigma_a}_b\}_{b \in B}$). Let the output assignment be $\varphi^{a, \sigma_a}$.
  \end{enumerate}
\item Execute the following greedy algorithm:
  \begin{enumerate}
  \item For each $a \in A$, assign $\sigma^*_a \in S_a$ to it that maximizes $\sum_{b \in B} \frac{1}{|S_b|}\left(\sum_{\sigma_b \in S_b} C_{(a, b)}(\sigma_a, \sigma_b)\right)$. \label{step:free-greedy}
  \item For each $b \in B$, assign $\sigma^*_b \in S_b$ to it that maximizes the number of edges satisfied, i.e., maximizes $\sum_{a \in A} C_{(a, b)}(\sigma^*_a, \sigma_b)$.
  \end{enumerate}
\item Output an assignment among the greedy assignment and $\varphi^{a, \sigma_a}$ for every $a, \sigma_a$ that satisfies maximum number of edges.
\end{enumerate}

Since every step except the {\sc Approx-FreeGame$_{j}$}($q, A, B, \{C_{(a, b)}\}_{(a, b) \in A \times B}, \{S_b\}_{b \in B}$) calls takes $O((nq)^2)$ time and we call {\sc Approx-FreeGame$_{j}$} only at most $(nq)^2$ times, we can conclude that the running time of {\sc Approx-FreeGame$_{j+1}$} is $O((nq)^{2j+2})$ as desired.

Define $R$ to be $n'\left(\sum_{b \in B} \left(\frac{d^{OPT}_b}{n'}\right)^\frac{j+2}{2}\left(\frac{1}{|S_b|}\right)^\frac{1}{j+1}1_{\sigma_b^{OPT} \in S_b}\right)$, our target number of edges we want to satisfy. The only thing left to show is that the assignment output from the algorithm indeed satisfies at least $R$ edges. We will consider two cases.

First, if there exist $a \in A$ and $\sigma_a \in S_b$ such that the output assignment from {\sc Approx-FreeGame$_{j}$}($q, A, B, \{C_{(a, b)}\}_{(a, b) \in A \times B}, \{S^{a, \sigma_a}_b\}_{b \in B}$) satisfies at least $R$ edges, then it is obvious that the output assignment of {\sc Approx-FreeGame$_{j+1}$} indeed satisfies at least $R$ edges as well.

In the second case, for every $a \in A$ and $\sigma_a \in S_a$, the output assignment from {\sc Approx-FreeGame$_{j}$}($q, A, B, \{C_{(a, b)}\}_{(a, b) \in A \times B}, \{S^{a, \sigma_a}_b\}_{b \in B}$) satisfies less than $R$ edges. For each $a \in A$, since the output assignment from {\sc Approx-FreeGame$_{j}$}($q, A, B, \{C_{(a, b)}\}_{(a, b) \in A \times B}, \{S^{a, \sigma^{OPT}_a}_b\}_{b \in B}$) satisfies less than $R$ edges, we arrive at the following inequality:
\begin{align*}
  \hspace*{1.5cm} R &> n'\left(\sum_{b \in B} \left(\frac{d^{OPT}_b}{n'}\right)^\frac{j+1}{2}\left(\frac{1}{|S^{a, \sigma^{OPT}_a}_b|}\right)^\frac{1}{j}1_{\sigma_b^{OPT} \in S^{a, \sigma^{OPT}_a}_b}\right) \\
  \hspace*{1.5cm} &\geq n'\left(\sum_{b \in \Gamma^{OPT}(a)} \left(\frac{d^{OPT}_b}{n'}\right)^\frac{j+1}{2}\left(\frac{1}{|S^{a, \sigma^{OPT}_a}_b|}\right)^\frac{1}{j}1_{\sigma_b^{OPT} \in S^{a, \sigma^{OPT}_a}_b}\right).
\end{align*}

Now, observe that, for every $b \in \Gamma^{OPT}(a)$, we have $1_{\sigma_b^{OPT} \in S^{a, \sigma^{OPT}_a}_b} = 1_{\sigma_b^{OPT} \in S_b}$. This is because, from our definition of $\Gamma^{OPT}$, $C_{(a, b)}(\sigma^{OPT}_a, \sigma^{OPT}_b) = 1$ for every $b \in \Gamma^{OPT}(a)$, which means that, if $\sigma^{OPT}_b$ is in $S_b$, then it remains in $S^{a, \sigma^{OPT}_a}_b$. Thus, the above inequality can be written as follows:
\begin{align} \label{ine:1}
  \hspace*{1.5cm} R &> n'\left(\sum_{b \in \Gamma^{OPT}(a)} \left(\frac{d^{OPT}_b}{n'}\right)^\frac{j+1}{2}\left(\frac{1}{|S^{a, \sigma^{OPT}_a}_b|}\right)^\frac{1}{j}1_{\sigma_b^{OPT} \in S_b}\right).
\end{align}

We will use inequality~(\ref{ine:1}) later in the proof. For now, we will turn our attention to the number of edges satisfied by the greedy algorithm, which, from our choice of $\sigma^*_b$, can be bounded as follows:
\begin{align*}
  \sum_{b \in B} \sum_{a \in A} C_{(a, b)}(\sigma^*_a, \sigma^*_b) &\geq \sum_{b \in B} \frac{1}{|S_b|} \sum_{\sigma_b \in S_b} \left(\sum_{a \in A} C_{(a, b)}(\sigma^*_a, \sigma_b)\right) \\
  &= \sum_{a \in A} \sum_{b \in B} \frac{1}{|S_b|} \left(\sum_{\sigma_b \in S_b} C_{(a, b)}(\sigma^*_a, \sigma_b)\right) \\
  (\text{From our choice of } \sigma_a^*) &\geq \sum_{a \in A} \sum_{b \in B} \frac{1}{|S_b|} \left(\sum_{\sigma_b \in S_b} C_{(a, b)}(\sigma^{OPT}_a, \sigma_b)\right) \\
  (\text{Since } C_{(a, b)}(\sigma^{OPT}_a, \sigma_b) = 1 \text{ for every } \sigma_b \in S^{a, \sigma^{OPT}_a}_b)&\geq \sum_{a \in A} \sum_{b \in B} \frac{1}{|S_b|} |S^{a, \sigma^{OPT}_a}_b| \\
  &= \sum_{a \in A} \sum_{b \in B} \frac{|S^{a, \sigma^{OPT}_a}_b|}{|S_b|} \\
  &\geq \sum_{a \in A} \sum_{b \in \Gamma^{OPT}(a)} \frac{|S^{a, \sigma^{OPT}_a}_b|}{|S_b|}.
\end{align*}

Moreover, from inequality~(\ref{ine:1}), we can derive the following inequalities:
\begin{align*}
  \hspace*{-1.5cm} &R^j\left(\sum_{a \in A} \sum_{b \in \Gamma^{OPT}(a)} \frac{|S^{a, \sigma^{OPT}_a}_b|}{|S_b|}\right) \\
  \hspace*{-1.5cm} &= \sum_{a \in A} R^j \left(\sum_{b \in \Gamma^{OPT}(a)} \frac{|S^{a, \sigma^{OPT}_a}_b|}{|S_b|}\right) \\
  \hspace*{-1.5cm} (\text{From~(\ref{ine:1}})) &\geq (n')^j \sum_{a \in A} \left(\sum_{b \in \Gamma^{OPT}(a)} \left(\frac{d^{OPT}_b}{n'}\right)^\frac{j+1}{2}\left(\frac{1}{|S^{a, \sigma^{OPT}_a}_b|}\right)^\frac{1}{j}1_{\sigma_b^{OPT} \in S_b}\right)^j \left(\sum_{b \in \Gamma^{OPT}(a)} \frac{|S^{a, \sigma^{OPT}_a}_b|}{|S_b|}\right) \\
  \hspace*{-1.5cm} (\text{H\"{o}lder's inequality}) &\geq (n')^j \sum_{a \in A} \left(\sum_{b \in \Gamma^{OPT}(a)} \left(\frac{d^{OPT}_b}{n'}\right)^\frac{j}{2}\left(\frac{1}{|S_b|}\right)^\frac{1}{j+1}1_{\sigma_b^{OPT} \in S_b}\right)^{j+1}
\end{align*}

By applying H\"{o}lder's inequality once again, the last term above is at least
\begin{align*}
  &(n')^j n' \left(\frac{1}{n'} \sum_{a \in A} \sum_{b \in \Gamma^{OPT}(a)} \left(\frac{d^{OPT}_b}{n'}\right)^\frac{j}{2}\left(\frac{1}{|S_b|}\right)^\frac{1}{j+1}1_{\sigma_b^{OPT} \in S_b}\right)^{j+1} \\
  &= \left(\sum_{b \in B} \sum_{a \in \Gamma^{OPT}(b)} \left(\frac{d^{OPT}_b}{n'}\right)^\frac{j}{2}\left(\frac{1}{|S_b|}\right)^\frac{1}{j+1}1_{\sigma_b^{OPT} \in S_b}\right)^{j+1} \\
  (\text{Since } d_b^{OPT} = |\Gamma^{OPT}(b)|) &= \left(\sum_{b \in B} d^{OPT}_b \left(\frac{d^{OPT}_b}{n'}\right)^\frac{j}{2}\left(\frac{1}{|S_b|}\right)^\frac{1}{j+1}1_{\sigma_b^{OPT} \in S_b}\right)^{j+1} \\
  &= \left(n' \sum_{b \in B} \left(\frac{d^{OPT}_b}{n'}\right)^\frac{j+2}{2}\left(\frac{1}{|S_b|}\right)^\frac{1}{j+1}1_{\sigma_b^{OPT} \in S_b}\right)^{j+1} \\
  &= R^{j+1}.
\end{align*}

Hence, we can conclude that $$\sum_{a \in A} \sum_{b \in \Gamma^{OPT}(a)} \frac{|S^{a, \sigma^{OPT}_a}_b|}{|S_b|} \geq R.$$ In other words, our greedy algorithm satisfies at least $R$ edges, which means that $P(j + 1)$ is also true for this second case.

As a result, $P(i)$ is true for every positive integer $i$, which completes the proof for Lemma~\ref{lem:main}.
\end{proof}

\section{Approximation Algorithm for Projection Games}

In this section, we will present our approximation algorithm for projection games. The main idea of this algorithm is a reduction from projection games on dense random graphs to free games, which we use together with the approximation algorithm for free games from Corollary~\ref{cor:free-game} above to prove Theorem~\ref{thm:dense-random-proj}. The reduction's properties can be stated formally as follows. \\

\begin{lemma} \label{lem:reduction-dense}
There is a polynomial-time reduction from a satisfiable projection game \\ $(q, A, B, E, \{\pi_e\}_{e \in E})$ where $(A, B, E)$ is sampled from a distribution $\mathcal{G}(n/2, n/2, p)$ where $p \geq 10\sqrt{\log n / n}$ to a satisfiable free game instance $(q', A', B', \{C_{(a, b)}\}_{(a, b) \in A' \times B'})$ such that, with probability $1 - o(1)$,
\begin{enumerate} \itemsep0em
  \item $|A'|, |B'| \leq |A|$ and $q' \leq q$, and
  \item For any $1 \geq \varepsilon \geq 0$, given an assignment $\varphi': A' \cup B' \to [q']$ to the free game instance of value $\varepsilon$, one can construct an assignment $\varphi: A \cup B \to [q]$ for the projection game of value $\Omega(\varepsilon)$ in polynomial time.
\end{enumerate}
\end{lemma}

Before we describe the reduction, we give a straightforward proof for Theorem~\ref{thm:dense-random-proj} based on the above lemma.

\begin{proof}[Proof of Theorem~\ref{thm:dense-random-proj} based on Lemma~\ref{lem:reduction-dense}]
The proof is simple. First, we use the reduction from Lemma~\ref{lem:reduction-dense} to transform a projection game on dense graph to a free game. Since the approximation ratio deteriorates by only constant factor with probability $1 - o(1)$ in the reduction, we can use the approximation algorithm from Corollary~\ref{cor:free-game} with $\lambda = 1$, which gives us an assignment of value at least $\Omega(1/N^\gamma)$.
\end{proof}

To prove the reduction lemma, we use the following two properties of random graphs. We do not prove the lemmas as they follow from a standard Chernoff bound. \\

\begin{lemma} \label{lem:prop-dense-random1}
When $p \geq 10\sqrt{\log n / n}$, with probability $1 - o(1)$, every vertex in $G \sim \mathcal{G}(n/2, n/2, p)$ has degree between $np/10$ and $10np$. \\
\end{lemma}

\begin{lemma} \label{lem:prop-dense-random2}
In $G \sim \mathcal{G}(n/2, n/2, p)$ with $p \geq 10\sqrt{\log n / n}$, with probability $1 - o(1)$, every pair of vertices $a, a'$ on the left has at least $np^2/10$ common neighbors.
\end{lemma}

Now, we are ready to prove the reduction lemma. Roughly speaking, the idea of the proof is to ``square'' the projection game, i.e., use $A$ as the vertices of the new game and, for each pair of vertices in $A$, add a constriant between them based on their constraints with their common neighbors in the projection game. This can be formalized as follows.

\begin{proof}[Proof of Lemma~\ref{lem:reduction-dense}]
The reduction proceeds as follows.
\begin{enumerate} \itemsep0em
\item Partition $A$ into $A_1, A_2$ of equal sizes. Then, set $A' \leftarrow A_1, B' \leftarrow A_2$ and $q' \leftarrow q$.
\item For each $a_1 \in A_1, a_2 \in A_2, \sigma_{a_1}, \sigma_{a_2} \in [q]$, let $C_{(a_1, a_2)}(\sigma_{a_1}, \sigma_{a_2})$ to be one if and only if these two assignments agree on every $b \in \Gamma(a_1) \cap \Gamma(a_2)$. In other words, $C_{(a_1, a_2)}(\sigma_{a_1}, \sigma_{a_2}) = 1$ if and only if $\pi_{(a_1, b)}(\sigma_{a_1}) = \pi_{(a_2, b)}(\sigma_{a_2})$ for every $b \in \Gamma(a_1) \cap \Gamma(a_2)$.
\end{enumerate}

It is obvious that the reduction runs in polynomial time, the first condition holds, and the new game is satisfiable. Thus, we only need to prove that, with probability $1 - o(1)$, the second condition is indeed true.

To show this, we present a simple algorithm that, given an assignment $\varphi': A' \cup B' \to [q']$ of the free game instance of value $\varepsilon$, output an assignment $\varphi: A \cup B \to [q]$ of the projection game of value $\Omega(\varepsilon)$. The algorithm works greedily as follows.
\begin{enumerate} \itemsep0em
\item For each $a \in A$, let $\varphi(a) \leftarrow \varphi'(a)$.
\item For each $b \in B$, pick $\varphi(b) = \sigma^*_b$ to be the assignment to $b$ that satisfies maximum number of edges, i.e., maximize $|\{a \in \Gamma(b) \mid \pi_{(a, b)}(\varphi(a)) = \sigma_b\}|$.
\end{enumerate}

Trivially, the algorithm runs in polynomial time. Thus, we only need to prove that, with probability $1 - o(1)$, the produced assignment is of value at least $\Omega(\varepsilon)$. To prove this, we will use the properties from Lemma~\ref{lem:prop-dense-random1} and Lemma~\ref{lem:prop-dense-random2}, which holds with probability $1 - o(1)$.

The number of satisfied edges can be written as follows.
\begin{align*}
  &\sum_{b \in B} \sum_{a \in \Gamma(b)} 1_{\pi_{(a, b)}(\varphi(a)) = \varphi(b)} = \sum_{b \in B} \sum_{a \in \Gamma(b)} 1_{\pi_{(a, b)}(\varphi'(a)) = \sigma^*_b}. \\
\end{align*}

Let $d_u$ be the degree of $u$ in $(A, B, E)$ for every $u \in A \cup B$, i.e. $d_u = |\Gamma(u)|$. We can further rearrange the above expression as follows.

\begin{align*}
\sum_{b \in B} \sum_{a \in \Gamma(b)} 1_{\pi_{(a, b)}(\varphi'(a)) = \sigma^*_b}
  &= \sum_{b \in B} \left[\frac{1}{d_b} \left(\sum_{a \in \Gamma(b)} 1_{\pi_{(a, b)}(\varphi'(a)) = \sigma^*_b}\right) d_b\right] \\
  &= \sum_{b \in B} \left[\frac{1}{d_b} \left(\sum_{a \in \Gamma(b)} 1_{\pi_{(a, b)}(\varphi'(a)) = \sigma^*_b}\right)\left(\sum_{a \in \Gamma(b)} 1\right)\right] \\
  &= \sum_{b \in B} \left[\frac{1}{d_b} \left(\sum_{a \in \Gamma(b)} 1_{\pi_{(a, b)}(\varphi'(a)) = \sigma^*_b}\right)\left(\sum_{a \in \Gamma(b)} \sum_{\sigma_b \in [q]} 1_{\pi_{(a, b)}(\varphi'(a)) = \sigma_b}\right)\right] \\
  &= \sum_{b \in B} \left[\frac{1}{d_b} \left(\sum_{a \in \Gamma(b)} 1_{\pi_{(a, b)}(\varphi'(a)) = \sigma^*_b}\right)\left(\sum_{\sigma_b \in [q]} \sum_{a \in \Gamma(b)} 1_{\pi_{(a, b)}(\varphi'(a)) = \sigma_b}\right)\right] \\
  &= \sum_{b \in B} \left[\frac{1}{d_b} \sum_{\sigma_b \in [q]} \left(\sum_{a \in \Gamma(b)} 1_{\pi_{(a, b)}(\varphi'(a)) = \sigma^*_b}\right)\left(\sum_{a \in \Gamma(b)} 1_{\pi_{(a, b)}(\varphi'(a)) = \sigma_b}\right)\right] \\
  (\text{From the choice of } \sigma^*_b)&\geq \sum_{b \in B} \left[\frac{1}{d_b} \sum_{\sigma_b \in [q]} \left(\sum_{a \in \Gamma(b)} 1_{\pi_{(a, b)}(\varphi'(a)) = \sigma_b}\right)^2\right] \\
  &= \sum_{b \in B} \left(\frac{1}{d_b} \sum_{\sigma_b \in [q]} \sum_{a, a' \in \Gamma(b)} 1_{\pi_{(a, b)}(\varphi'(a)) = \sigma_b}1_{\pi_{(a', b)}(\varphi'(a')) = \sigma_b}\right) \\
  &= \sum_{b \in B} \left(\frac{1}{d_b} \sum_{a, a' \in \Gamma(b)} \sum_{\sigma_b \in [q]} 1_{\pi_{(a, b)}(\varphi'(a)) = \sigma_b}1_{\pi_{(a', b)}(\varphi'(a')) = \sigma_b}\right)
\end{align*}

Observe that $\sum_{\sigma_b \in [q]} 1_{\pi_{(a, b)}(\varphi'(a)) = \sigma_b}1_{\pi_{(a', b)}(\varphi'(a')) = \sigma_b} = 1_{\pi_{(a, b)}(\varphi'(a)) = \pi_{(a', b)}(\varphi'(a'))}$. Thus, the number of satisfied edges is at least
$$\sum_{b \in B} \left(\frac{1}{d_b} \sum_{a, a' \in \Gamma(b)} 1_{\pi_{(a, b)}(\varphi'(a)) = \pi_{(a', b)}(\varphi'(a'))}\right).$$

Moreover, from Lemma~\ref{lem:prop-dense-random1}, $d_b \leq 10np$ for every $b \in B$ with probability $1 - o(1)$. This implies that, with probability $1 - o(1)$, the output assignment satisfied at least $$\frac{1}{10np} \sum_{b \in B} \sum_{a, a' \in \Gamma(b)} 1_{\pi_{(a, b)}(\varphi'(a)) = \pi_{(a', b)}(\varphi'(a'))}$$ edges.

We can further reorganize this quantity as follows.
\begin{align*}
  \frac{1}{10np} \sum_{b \in B} \sum_{a, a' \in \Gamma(b)} 1_{\pi_{(a, b)}(\varphi'(a)) = \pi_{(a', b)}(\varphi'(a'))} &\geq \frac{1}{10np} \sum_{b \in B} \sum_{{(a, a') \in A' \times B' \atop \text{s.t. } a, a' \in \Gamma(b)}} 1_{\pi_{(a, b)}(\varphi'(a)) = \pi_{(a', b)}(\varphi'(a'))} \\
  &= \frac{1}{10np} \sum_{(a, a') \in A' \times B'} \sum_{b \in \Gamma(a) \cap \Gamma(a')} 1_{\pi_{(a, b)}(\varphi'(a)) = \pi_{(a', b)}(\varphi'(a'))}.
\end{align*}

Now, observe that, from its definition, if $C_{(a, a')}(\varphi'(a), \varphi'(a'))$ is one, then $1_{\pi_{(a, b)}(\varphi'(a)) = \pi_{(a', b)}(\varphi'(a'))}$ is also one for every $b \in \Gamma(a) \cap \Gamma(a')$. Thus, we have
\begin{align*}
  \hspace{1.5cm} &\frac{1}{10np} \sum_{(a, a') \in A' \times B'} \sum_{b \in \Gamma(a) \cap \Gamma(a')} 1_{\pi_{(a, b)}(\varphi'(a)) = \pi_{(a', b)}(\varphi'(a'))} \\
  \hspace{1.5cm} &\geq \frac{1}{10np} \sum_{(a, a') \in A' \times B'} \sum_{b \in \Gamma(a) \cap \Gamma(a')} C_{(a, a')}(\varphi'(a), \varphi'(a')) \\
  \hspace{1.5cm} &= \frac{1}{10np} \sum_{(a, a') \in A' \times B'} |\Gamma(a) \cap \Gamma(a')| C_{(a, a')}(\varphi'(a), \varphi'(a')).
\end{align*}

From Lemma~\ref{lem:prop-dense-random2}, with probability $1 - o(1)$, $|\Gamma(a) \cap \Gamma(a')| \geq np^2/10$ for every $(a, a') \in A' \times B'$. Hence, we can conclude that the above expression is, with probability $1 - o(1)$, at least
\begin{align*}
  \frac{1}{10np} \sum_{(a, a') \in A' \times B'} \frac{np^2}{10} C_{(a, a')}(\varphi'(a), \varphi'(a')) &= \frac{p}{100} \sum_{(a, a') \in A' \times B'} C_{(a, a')}(\varphi'(a), \varphi'(a')).
\end{align*}

Next, note that $\sum_{(a, a') \in A' \times B'} C_{(a, a')}(\varphi'(a), \varphi'(a'))$ is the number of edges satisfied by $\varphi'$ in the free game, which is at least $\varepsilon |A'||B'| = \varepsilon n^2/16$. Thus, we have
$$  \frac{p}{100} \sum_{(a, a') \in A' \times B'} C_{(a, a')}(\varphi'(a), \varphi'(a')) \geq \frac{\varepsilon n^2p}{1600}.$$
Finally, again from Lemma~\ref{lem:prop-dense-random1}, the total number of edges is at most $5n^2p$ with probability $1 - o(1)$. As a result, with probability $1 - o(1)$, the algorithm outputs an assignment that satisfies at least $\frac{\varepsilon}{8000} = \Omega(\varepsilon)$ fraction of edges of the projection game instance as desired.
\end{proof}

\section{Approximation Algorithm for Densest $k$-Subgraph}

The main goal of this section is to prove Corollary~\ref{cor:dks}. As stated previously, we simply use our algorithm from Theorem~\ref{thm:main} together with a reduction from {\sc Max 2-CSP} to {\sc D$k$S} from~\cite{CHK}. First, let us start by stating the reduction from Theorem~\ref{thm:main}, which we rephrase as follows. \\

\begin{lemma}[\cite{CHK}] \label{lem:dks-reduction}
  There exists a randomized polynomial-time algorithm that, given a graph $G$ of $N$ vertices and an integer $k \leq N$, produces an instance $(q, V, E, \{C_e\}_{e \in E})$ of {\sc Max 2-CSP} such that
  \begin{itemize}
  \item $q \leq N, n = k$, and
  \item any solution to the instance can be translated in polynomial time to a subgraph of $G$ of $k$ vertices such that the number of edges in the subgraph equals to the number of edges satisfied by the {\sc Max 2-CSP} solution, and
  \item with constant probability, the number of edges satisfied by the optimal solution to the instance is at least $1/100$ times the number of edges in the densest $k$-subgraph of $G$.
  \end{itemize}
\end{lemma}

We will not show the proof of Lemma~\ref{lem:dks-reduction} here; please refer to Theorem 6 from~\cite{CHK} for the proof. Instead, we will now show how to use the reduction to arrive at the proof of Corollary~\ref{cor:dks}.

\begin{proof}[Proof of Corollary~\ref{cor:dks}]
  First, we note that, to prove Corollary~\ref{cor:dks}, it is enough to find a randomized polynomial-time algorithm with similar approximation guarantee to that in Corollary~\ref{cor:dks} except that the probability of success is a constant (instead of high probability as stated in Corollary~\ref{cor:dks}). This is because we can then repeatedly run this algorithm $\Theta(\log n)$ times and produce the desired result.

  The algorithm proceeds as follows:
  \begin{enumerate}
    \item Use the reduction from Lemma~\ref{lem:dks-reduction} on the input graph $G$ and $k$ to produce $(q, V, E, \{C_e\}_{e \in E})$.
    \item Run the algorithm from Theorem~\ref{thm:main} on $(q, V, E, \{C_e\}_{e \in E})$. \label{step:dks-main-csp}
    \item Transform the assignment from previous step according to Lemma~\ref{lem:dks-reduction} and output the result.
  \end{enumerate}

  From the property of the reduction, we know that, with constant probability, the optimal assignment to $(q, V, E, \{C_e\}_{e \in E})$ satisfies $\Omega(\delta k^2)$ edges. If this is the case, we can conclude that the density of $(V, E)$ is $\Omega(\delta)$ and, similarly, that the value of the instance is $\Omega(\delta)$. As a result, the output assignment from step~\ref{step:dks-main-csp} has value at least $\Omega(\delta^{O(1/\gamma)}N^{-\gamma})$. Since the reduction from Lemma~\ref{lem:dks-reduction} preserves the optimum, our algorithm produces a subgraph of density at least $\Omega(\delta^{O(1/\gamma)}N^{-\gamma})$ as well, which concludes our proof for this corollary.
\end{proof}

\section{QPTAS for Dense Max 2-CSPs}

At first glance, it seems that the QPTAS would follow easily for our main theorem. This, however, is not the case as the algorithm in the main theorem always loses at least a constant factor. Instead, we need to give an algorithm that is similar to that of the main theorem but have a stronger guarantee in approximation ratio for satisfiable instances, which can be stated as follows. \\

\begin{lemma} \label{lem:approx-complete-game}
  For every positive integer $i > 0$, there exists an $O\left((nq)^{O(i)}\right)$-time algorithm that, for any satisfiable {\sc Max 2-CSP} instance on the complete graph, produces an assignment of value at least $1/q^{1/i}$.
\end{lemma}

Lemma~\ref{lem:approx-complete-game} can be viewed as a special case of the main theorem when the graph is complete. However, it should be noted that Lemma~\ref{lem:approx-complete-game} is more exact in the sense that the guaranteed lower bound of the value of the output assignment is not asymptotic. The proof of this lemma is also similar to that of Lemma~\ref{lem:main} except that we need slightly more complicated algorithm and computation to deal with the fact that the underlying graph is not bipartite.

\begin{proof}[Proof of Lemma~\ref{lem:approx-complete-game}]
We will prove the lemma by induction. Note that throughout the proof, we will not worry about the randomness that the algorithm employs; it is not hard to see that the random assignment algorithms described below can be derandomized via greedy approach so that the approximation guarantees are as good as the expected guarantees of the randomized ones and that we still end up with the same asymptotic running time.

Let $P(i)$ represent the following statement: there exists an $O\left((nq)^{3i}\right)$-time algorithm {\sc Approx-CompleteGame$_i$}($q, V, E, \{C_{e}\}_{e \in E}, \{S_u\}_{u \in V}$) that takes in a satisfiable {\sc Max 2-CSP} instance $(q, V, E, \{C_{e}\}_{e \in V})$ where $(V, E)$ is a complete graph and a reduced alphabet set $S_u$ for every $u \in U$ such that, if $\sigma^{OPT}_u \in S_u$ for every $u \in V$, then the algorithm outputs an assignment of value at least $\left(\prod_{u \in V} \frac{1}{|S_u|}\right)^{\frac{1}{ni}}$.

Observe that $P(i)$ implies the lemma by simply setting $S_u = [q]$ for every $u \in V$.

{\em Base Case.} The algorithm {\sc Approx-CompleteGame$_1$}($q, V, E, \{C_{e}\}_{e \in E}, \{S_u\}_{u \in V}$) is a simple random assignment algorithm. However, before we randomly pick the assignment, we need to first discard the alphabets that we know for sure are not optimal. More specifically, {\sc Approx-CompleteGame$_1$}($q, V, E, \{C_{e}\}_{e \in E}, \{S_u\}_{u \in V}$) works as follows.
\begin{enumerate}
  \item While there exist $u, v \in U$ and $\sigma_u \in S_u$ such that $C_{(u, v)}(\sigma_u, \sigma_v) = 0$ for every $\sigma_v \in S_v$, remove $\sigma_u$ from $S_u$. \label{step:remove}
  \item For each $u \in V$, pick $\varphi(u)$ independently and uniformly at random from $S_u$. Output $\varphi$.
\end{enumerate}

It is obvious that the algorithm runs in $O(n^3q^3)$ time as desired.

Now, we will show that, if $\sigma_u^{OPT} \in S_u$ for every $u \in V$, then the algorithm gives an assignment that is of value at least $\left(\prod_{u \in V} \frac{1}{|S_u|}\right)^{\frac{1}{n}}$ in expectation.

First, observe that $\sigma_u^{OPT}$ remains in $S_u$ after step~\ref{step:remove} for every $u \in V$. This is because $C_{(u, v)}(\sigma_u^{OPT}, \sigma_v^{OPT}) = 1$ for every $v \neq u$.

Next, Consider the expected number of satisfied edges by the output assignment, which can be rearranged as follows:
\begin{align*}
  \mathbb{E}\left[\sum_{(u, v) \in E} C_{(u, v)}(\varphi(u), \varphi(v))\right] &= \sum_{(u, v) \in E} \mathbb{E}\left[C_{(u, v)}(\varphi(u), \varphi(v))\right] \\
  &= \sum_{(u, v) \in E} \frac{1}{|S_u||S_v|} \sum_{\sigma_u \in S_u} \sum_{\sigma_v \in S_v} C_{(u, v)}(\sigma_u, \sigma_v). \\
\end{align*}

From the condition of the loop in step~\ref{step:remove}, we know that after the loop ends, for each $\sigma_u \in S_u$, there must be at least one $\sigma_v \in S_v$ such that $C_{(u, v)}(\sigma_u, \sigma_v) = 1$. In other words, $$\sum_{\sigma_u \in S_u} \sum_{\sigma_v \in S_v} C_{(u, v)}(\sigma_u, \sigma_v) \geq \sum_{\sigma_u \in S_u} 1 = |S_u|.$$

Similarly, we can also conclude that $$\sum_{\sigma_u \in S_u} \sum_{\sigma_v \in S_v} C_{(u, v)}(\sigma_u, \sigma_v) \geq |S_v|.$$

Thus, we have $$\sum_{\sigma_u \in S_u} \sum_{\sigma_v \in S_v} C_{(u, v)}(\sigma_u, \sigma_v) \geq \max\{|S_u|,|S_v|\}$$ for every $u \neq v$.

Hence, we can bound the expected number of satisfied edges as follows:
\begin{align*}
  \sum_{(u, v) \in E} \frac{1}{|S_u||S_v|} \sum_{\sigma_u \in S_u} \sum_{\sigma_v \in S_v} C_{(u, v)}(\sigma_u, \sigma_v)
  &\geq \sum_{(u, v) \in E} \frac{1}{|S_u||S_v|} \max\{|S_u|, |S_v|\} \\
  &= \sum_{(u, v) \in E} \frac{1}{\min\{|S_u|, |S_v|\}} \\
  &\geq \sum_{(u, v) \in E} \frac{1}{\sqrt{|S_u||S_v|}} \\
  (\text{A.M. - G.M. inequality}) &\geq |E| \left(\prod_{(u, v) \in E} \frac{1}{\sqrt{|S_u||S_v|}}\right)^{\frac{1}{|E|}} \\
  &= |E| \left(\prod_{(u, v) \in E} \frac{1}{\sqrt{|S_u||S_v|}}\right)^{\frac{2}{n(n-1)}} \\
  (\text{Each } u \in V \text{ appears in exactly } n - 1 \text{ edges}) &= |E| \left(\left(\prod_{u \in V} \frac{1}{|S_u|}\right)^{(n-1)/2}\right)^{\frac{2}{n(n-1)}} \\
  &= |E| \left(\prod_{u \in V} \frac{1}{|S_u|}\right)^{1/n},
\end{align*}
which implies that $P(1)$ is true as desired.

{\em Inductive Step.} Let $j$ be any positive integer. Suppose that $P(j)$ holds.

We will now describe {\sc Approx-CompleteGame$_{j+1}$} based on {\sc Approx-CompleteGame$_j$} as follows.

\begin{enumerate}
\item Define $R$ to be $\left(\prod_{u \in V} \frac{1}{|S_u|}\right)^{\frac{1}{n(j + 1)}}$, our target value we want to achieve.
\item Run the following steps~\ref{step:reduce} to~\ref{step:exc} until no $S_u$ is modified by neither step~\ref{step:exc} nor step~\ref{step:empty}.
  \begin{enumerate}
  \item For each $u \in V$ and $\sigma_u \in S_u$, do the following:
    \begin{enumerate}
    \item For each $v \in V$, compute $S_v^{u, \sigma_u} = \{\sigma_v \in S_v \mid C_{(u, v)}(\sigma_u, \sigma_v) = 1\}$. This is the set of reduced assignments of $v$ if we assign $\sigma_u$ to $u$. Note that when $v = u$, let $S_u = \{\sigma_u\}$. \label{step:reduce}
    \item If $S_v^{u, \sigma_u} = \emptyset$ for some $v \in V$, then remove $\sigma_u$ from $S_u$ and continue to the next $u, \sigma_u$ pair. \label{step:empty}
    \item Compute $R^{u, \sigma_u} = \left(\prod_{v \in V} \frac{1}{|S_v^{u, \sigma_u}|}\right)^{\frac{1}{nj}}$. If $R' < R$, continue to the next $u, \sigma_u$ pair.
    \item Execute {\sc Approx-CompleteGame$_j$}($q, V, E, \{C_{e}\}_{e \in E}, \{S_v^{u, \sigma_u}\}_{v \in V}$). If the output assignment is of value less than $R^{u, \sigma_u}$, then remove $\sigma_u$ from $S_u$. Otherwise, return the output assignment as the output to {\sc Approx-CompleteGame$_{j+1}$}. \label{step:exc}
    \end{enumerate}
  \end{enumerate}
\item If the loop in the previous step ends without outputting any assignment, just output a random assignment (i.e. pick $\varphi(u)$ independently and uniformly at random from $S_u$). \label{step:random-assignment}
\end{enumerate}

Observe first that the loop can run at most $nq$ times as the total number of elements of $S_v$'s for all $v \in V$ is at most $nq$. This means that we call {\sc Approx-CompleteGame$_{j}$} at most $nq$ times. Since every step except the {\sc Approx-CompleteGame$_{j}$} calls takes $O((nq)^3)$ time and we call {\sc Approx-CompleteGame$_{j}$} only at most $n^2q^2$ times, we can conclude that the running time of {\sc Approx-CompleteGame$_{j+1}$} is $O((nq)^{3j+3})$ as desired.

The only thing left to show is that the assignment output from the algorithm indeed is of expected value at least $R$. To do so, we will consider two cases.

First, if step~\ref{step:random-assignment} is never reached, the algorithm must terminate at step~\ref{step:exc}. From the return condition in step~\ref{step:exc}, we know that the output assignment is of value at least $R^{u, \sigma_u} \geq R$ as desired.

In the second case where step~\ref{step:random-assignment} is reached, we first observe that when we remove $s_u$ from $S_u$ in step~\ref{step:exc}, the instance is still satisfiable. The reason is that, if $\sigma_u = \sigma_u^{OPT}$ is the optimal assignment for $u$, then $\sigma_v^{OPT}$ remains in $S_v^{u, \sigma_u}$ for every $v \in V$. Hence, from our inductive hypothesis, the output assignment from {\sc Approx-CompleteGame$_j$}($q, V, E, \{C_{e}\}_{e \in E}, \{S_v^{u, \sigma_u}\}_{v \in V}$) must be of value at least $R^{u, \sigma_u}$. As a result, we never remove $s^{OPT}_u$ from $S_u$, and, thus, the instance remains satisfiable throughout the algorithm.

Moreover, notice that, if $R^{u, \sigma_u} \geq R$ for any $u, \sigma_u$, we either remove $\sigma_u$ from $S_u$ or output the desired assignment. This means that, when step~\ref{step:random-assignment} is reached, $R^{u, \sigma_u} < R$ for every $u \in V$ and $\sigma_u \in S_u$.

Now, let us consider the expected number of edges satisfied by the random assignment. Since our graph $(V, E)$ is complete, it can be written as follows.

\begin{align*}
  \mathbb{E}\left[\sum_{(u, v) \in E} C_{(u, v)}(\varphi(u), \varphi(v))\right] &= \mathbb{E}\left[\frac{1}{2}\sum_{u \in V} \sum_{v \in V \atop v \neq u} C_{(u, v)}(\varphi(u), \varphi(v))\right] \\
  &= \frac{1}{2}\sum_{u \in V} \sum_{{v \in V \atop v \neq u}} \mathbb{E}\left[C_{(u, v)}(\varphi(u), \varphi(v))\right] \\
  &= \frac{1}{2}\sum_{u \in V} \sum_{{v \in V \atop v \neq u}} \frac{1}{|S_u||S_v|} \left(\sum_{\sigma_u \in S_u} \sum_{\sigma_v \in S_v} C_{(u, v)}(\sigma_u, \sigma_v)\right) \\
  (\text{From definition of } S_v^{u, \sigma_u}) &= \frac{1}{2}\sum_{u \in V} \sum_{{v \in V \atop v \neq u}} \frac{1}{|S_u||S_v|} \left(\sum_{\sigma_u \in S_u} |S_v^{u, \sigma_u}|\right) \\
  &= \frac{1}{2} \sum_{u \in V} \frac{1}{|S_u|} \sum_{\sigma_u \in S_u} \left(\sum_{{v \in V \atop v \neq u}} \frac{|S_v^{u, \sigma_u}|}{|S_v|}\right) \\
  (\text{A.M.-G.M. inequality}) &\geq \frac{1}{2} \sum_{u \in V} \frac{1}{|S_u|} \sum_{\sigma_u \in S_u} (n - 1)\sqrt[n-1]{\prod_{{v \in V \atop v \neq u}} \frac{|S_v^{u, \sigma_u}|}{|S_v|}} \\
 &= \frac{(n - 1)}{2} \sum_{u \in V} \frac{1}{|S_u|} \sum_{\sigma_u \in S_u} \sqrt[n-1]{\frac{\prod_{{v \in n \atop v \neq u}} |S_v^{u, \sigma_u}|}{\prod_{{v \in V \atop v \neq u}} |S_v|}} \\
 (\text{From our definition of } R^{u, \sigma_u}, R) &= \frac{(n - 1)}{2} \sum_{u \in V} \frac{1}{|S_u|} \sum_{\sigma_u \in S_u} \sqrt[n-1]{\frac{(R^{u, \sigma_u})^{-nj}}{\frac{R^{-n(j + 1)}}{|S_u|}}} \\
 (\text{Since } R^{u, \sigma_u} < R) &> \frac{(n - 1)}{2} \sum_{u \in V} \frac{1}{|S_u|} \sum_{\sigma_u \in S_u} \sqrt[n-1]{|S_u|R^{n}} \\
 &= \frac{(n - 1)}{2} \sum_{u \in V} \sqrt[n-1]{|S_u|R^{n}} \\
 &= \frac{(n - 1)}{2} R^{n/(n-1)} \left(\sum_{u \in V} \sqrt[n-1]{|S_u|} \right) \\
 (\text{A.M.-G.M. inequality}) &\geq \frac{(n-1)}{2} R^{n/(n-1)} \left(n\sqrt[n(n-1)]{\prod_{u \in V} |S_u|}\right) \\
 (\text{From our definition of } R) &= \frac{(n-1)}{2} R^{n/(n-1)} \left(n\sqrt[n(n-1)]{\prod_{u \in V} R^{-n(j + 1)}}\right) \\
 &= \frac{n(n-1)}{2} R^{(n-1-j)/(n-1)} \\
 (\text{Since } R \leq 1 \text{ and } j \geq 0) &\geq \frac{n(n-1)}{2} R.
\end{align*}

Since $\frac{n(n-1)}{2}$ is the number of edges in $(V, E)$, we can conclude that the random assignment is indeed of expected value at least $R$.

Thus, we can conclude that $P(j + 1)$ is true. As a result, $P(i)$ is true for every positive integer $i$, which completes the proof for Lemma~\ref{lem:approx-complete-game}.
\end{proof}

Next, we will prove Corollary~\ref{cor:qptas-dense} by reducing it to {\sc Max 2-CSP} on complete graph, and, then plug in Lemma~\ref{lem:approx-complete-game} with appropriate $i$ to get the result.

First, observe that, since $\log(1 + \varepsilon') = \Omega(\varepsilon')$ for every $1 \geq \varepsilon' > 0$, by plugging in $i = C \log q/\varepsilon'$ for large enough constant $C$ into Lemma~\ref{lem:approx-complete-game}, we immediately arrive the following corollary. \\

\begin{corollary} \label{cor:qptas-complete}
  For any $1 \geq \varepsilon' > 0$, there exists an $(1 + \varepsilon')$-approximation algorithm for satisfiable {\sc Max 2-CSP} on the complete graph that runs in time $N^{O(\varepsilon'^{-1}\log N)}$.
\end{corollary}

Now, we will proceed to show the reduction and, thus, prove Corollary~\ref{cor:qptas-dense}.

\begin{proof}[Proof of Corollary~\ref{cor:qptas-dense}]
  First of all, notice that, since $\frac{1}{1 + \varepsilon} = 1 - \Theta(\varepsilon)$. It is enough for us to show that there exists an $N^{O(\varepsilon^{-1}\delta^{-1}\log N)}$-time algorithm for satisfiable $\delta$-dense {\sc Max 2-CSP} that produces an assignment of value at least $1 - \varepsilon$.

  On input $(q, V, E, \{C_e\}_{e \in E})$, the algorithm works as follows:
  \begin{enumerate}
    \item Construct a {\sc Max 2-CSP} instance $(q, V, E', \{C'_e\}_{e \in E'})$ where $(V, E')$ is a complete graph and $C'_e$ is defined as $C_e$ if $e \in E$. Otherwise, $C_e := 1$. In other words, we put in dummy constraints that are always true just to make the graph complete.
    \item Run the algorithm from Corollary~\ref{cor:qptas-complete} on $(q, V, E', \{C'_e\}_{e \in E'})$ with $\varepsilon' = \varepsilon \delta$ and output the assignment got from the algorithm.
  \end{enumerate}

  To see that the algorithm indeed produces an assignment with value $1 - \varepsilon$ for the input instance, first observe that, since $(q, V, E, \{C_e\}_{e \in E})$ is satisfiable,  $(q, V, E', \{C'_e\}_{e \in E'})$ is trivially satisfiable. Thus, from Corollary~\ref{cor:qptas-complete}, the output assignment has value at least $1/(1 + \delta\varepsilon) \geq 1 - \delta \varepsilon$ with respect to $(q, V, E', \{C'_e\}_{e \in E'})$. In other words, the assignment does not satisfy at most $\delta \varepsilon n^2$ edges. Thus, with respect to the input instance, it satisfies at least $\delta n^2 - \delta \varepsilon n^2 = (1 - \varepsilon) \delta n^2$ edges. In other words, it is of value at least $1 - \varepsilon$ as desired.

  Lastly, note that the running time of this algorithm is determined by that of the algorithm from Corollary~\ref{cor:qptas-complete}, which runs in $N^{O(\varepsilon'^{-1}\log N)} = N^{O(\varepsilon^{-1}\delta^{-1}\log N)}$ time as desired.
\end{proof}

\section{Conclusions and Open Questions}
Finally, we conclude by listing the open questions and interesting directions related to the techniques and problems presented here. We also provide our thoughts regarding each question.
\begin{itemize}
\item {\em Can our algorithm be extended to work for {\sc Max $k$-CSP} for $k \geq 3$?} Other algorithms for approximating {\sc Max 2-CSP} such as those from~\cite{Alon:2003:RSA:963875.963877, AKK95, BMHS11} are applicable for {\sc Max $k$-CSP} for any value of $k$ as well. So it is possible that our technique can be employed for {\sc Max $k$-CSP} too.
\item {\em Can one also come up with an algorithm that approximates {\sc Max 2-CSP} to within $O(N^\varepsilon)$ factor for any $\varepsilon > 0$ for low-value dense {\sc Max 2-CSP}?} Our algorithm needs the value $\lambda$ to be $N^{-o(1)}$ in order to give such a ratio so it is interesting whether we can remove or relax this condition. However, we do not think that one can remove the condition completely because, with similar technique to the proof of Corollary~\ref{cor:qptas-dense}, we can arrive at a reduction from any {\sc Max 2-CSP} to dense {\sc Max 2-CSP} where the approximation ratio is preserved but the value decreases. This means that, if we can remove the condition on $\lambda$, then we are also able to refute the PGC. This argument nonetheless does not rule out relaxing the condition for $\lambda$ without removing it completely.
\item {\em Can our QPTAS be extended to unsatisfiable instances?} One of the main disadvantages of our QPTAS is that it requires the instance to be satisfiable. This renders our QPTAS useless against many problems such as {\sc Max $2$-SAT} and {\sc Max-Cut} because the satisfiable instances of those problems are trivial. If we can extend our QPTAS to work on unsatisfiable instances as well, then we may be able to produce interesting results for those problems.  Note, however, that, with similar argument to the preceding question, QPTAS for low-value instances likely does not exist. Instead, the case of unsatisfiable instances where~\cite{Alon:2003:RSA:963875.963877, AKK95, BMHS11} are successful is when they look for an additive error guarantee instead of a multiplicative one. Currently, it is unclear whether our technique can achieve such results.
\item {\em Can one arrive at a similar or even better algorithm using SDP hierarchies?} SDP hierarchies have been very useful in finding approximation algorithms for combinatorial optimization problems. A natural question to ask is whether one can apply SDP hierarchies to get similar results to ours. For example, can the $O(i)$-level of the Lasserre hierarchy produce an approximation algorithm with ratio $O(q^{1/i})$ for dense {\sc Max 2-CSP}? If so, then this may also be an interesting direction to pursue an algorithm with guarantee additive error discussed previously.
\end{itemize}

\bibliographystyle{plain}

\bibliography{bi}

\end{document}